\documentclass[english,runningheads]{llncs}
\usepackage[latin1]{inputenc}
\usepackage[T1]{fontenc}
\usepackage{babel}
\usepackage{graphicx}
\usepackage{amsmath}
\usepackage{amsfonts}
\usepackage{amssymb}
\usepackage{multirow}
\usepackage{multicol}
\usepackage{breakcites}
\usepackage{url}
\usepackage{stmaryrd}
\usepackage{lmodern}  
\usepackage{longtable}
\usepackage{rotating}
\usepackage{array}
\usepackage{pdflscape}
\usepackage{soul} 
\usepackage{xspace}
\usepackage{enumitem}
\usepackage{hyperref}
\usepackage{color}
\usepackage{wrapfig} 
\usepackage{colortbl} 

\usepackage{proof}
\usepackage{makecell}

\newcommand{\Diff}{${\mathcal R}$}

\newcommand{\If}{\leftarrow}

\newcommand{\vars}{\mathit{vars}}

\newcommand{\repl}{\! \Rightarrow \!}

\newcommand{\Embedded}
{\mbox{\hspace*{.5mm}\raisebox{-1.1mm}{$\sim$}\hspace*{-2.5mm}\raisebox{.2mm}{\large{$\triangleleft$}}\hspace*{.9mm}}}
\newcommand{\Down}{\raisebox{-.5mm}{\rule{0mm}{1mm}}}

\title{Removing Algebraic Data Types\\ from Constrained Horn Clauses\\ Using Difference Predicates\thanks{This work has been partially supported by the National Group of Computing Science (GNCS-INdAM). 
This paper has been published in:
Automated Reasoning,
10th International Joint Conference, IJCAR 2020, Paris, France, July 1--4, 2020, 
Proceedings, Lecture Notes in Artificial Intelligence no.12166 Part~I,
N.~Peltier and V.~Sofronie-Stokkermans (Eds.), 
pp.~83--102, 2020. Springer Nature.
The final authenticated publication is available online at \url{https://doi.org/10.1007/978-3-030-51074-9_6}}}

\titlerunning{Removing ADTs from CHCs using Difference Predicates}

\author{
Emanuele De Angelis\inst{1,3}\orcidID{0000-0002-7319-8439}\and
Fabio~Fioravanti\inst{1}\orcidID{0000-0002-1268-7829} \and
Alberto~Pettorossi\inst{2,3}\orcidID{0000-0001-7858-4032}\and
Maurizio~Proietti\inst{3}\orcidID{0000-0003-3835-4931}
}

\authorrunning{E.~De~Angelis, F.~Fioravanti, A.~Pettorossi, and M.~Proietti}

\institute{DEC, University `G. d'Annunzio', Chieti-Pescara, Italy\\
	\email{fabio.fioravanti@unich.it}\\
	\and DICII, University of Rome `Tor Vergata', Italy\\
	\email{pettorossi@info.uniroma2.it} \\
	\and IASI-CNR, Rome, Italy \\
	\email{\{emanuele.deangelis,maurizio.proietti\}@iasi.cnr.it}
}
\begin{document}
	\maketitle
	\begin{abstract}

We address the problem of proving the satisfiability of 
Constrained Horn Clauses (CHCs) with Algebraic Data Types (ADTs), 
such as lists and trees.
We propose a new technique for transforming CHCs with 
ADTs into CHCs where predicates are defined over basic types, such as
integers and booleans, only.
Thus, our technique avoids the explicit use of inductive proof rules
during satisfiability proofs.
The main extension over previous techniques for ADT removal is 
a new transformation rule,
called {\it differential replacement}, which allows us to introduce
auxiliary predicates corresponding to the lemmas used
when making inductive proofs.
We present an algorithm that applies the new 
rule, together with the traditional folding/unfolding rules, 
for the automatic removal of ADTs.
We prove that if the set of the transformed clauses is satisfiable, 
then so is the set of the original clauses.
By an experimental evaluation, 
we show that the use of the new rule 
significantly improves the effectiveness of ADT removal, and
that our approach is competitive 
with respect to a state-of-the-art tool that extends the CVC4 solver with induction.

\end{abstract}
	
	\section{Introduction}
	
	\label{sec:Intro}
	\label{intro}

{\it Constrained Horn Clauses} (CHCs) constitute a fragment of
the first order predicate calculus, where the Horn clause format is extended 
by allowing {\it constraints} on specific domains
to occur in clause premises. CHCs have gained popularity as 
a suitable logical formalism for automatic program verification~\cite{Bj&15}.
Indeed, 
many verification problems can be reduced to the satisfiability 
problem for CHCs.

Satisfiability of CHCs is a particular case of
{\it Satisfiability Modulo Theories} (SMT), understood here as the general 
problem of determining the satisfiability of (possibly quantified) 
first order formulas where the interpretation of some
function and predicate symbols is defined in
a given constraint (or {\it background}) theory~\cite{BaT18}.
Recent advances in the field have led to the development of a number of
very powerful SMT 
(and, in particular, CHC) {\it solvers}, which
aim at solving satisfiability problems with respect to a large
variety of constraint theories.
Among SMT solvers, we would like to mention
CVC4~\cite{CVC4}, MathSAT~\cite{MaS13}, and Z3~\cite{DeB08},
and among solvers with specialized engines for CHCs,
we recall Eldarica~\cite{HoR18}, HSF~\cite{Gr&12}, RAHFT~\cite{Ka&16}, and Spacer~\cite{Ko&13}.

{Even if SMT algorithms for unrestricted first order formulas 
suffer from incompleteness limitations due to general undecidability results,} 
most of the above mentioned tools work well {in practice} {when acting}
on constraint theories, such as 
Booleans, Uninterpreted Function Symbols, Linear Integer or Real Arithmetic, Bit Vectors, and
Arrays.
However, when formulas contain {universally quantified variables}
{ranging over} inductively defined {\it algebraic data types} (ADTs),
such as lists and trees, then 
{the SMT/CHC solvers often show a very poor performance, as they do not incorporate  
induction principles relative to the ADT in use.}

To {mitigate} this difficulty, some SMT/CHC solvers 
have been enhanced by incorporating appropriate
induction principles~\cite{ReK15,Un&17,Ya&19},
similarly to what has been done in automated theorem provers~\cite{Bun01}.
The most creative step which is needed when extending SMT solving 
with induction,
is the generation of the auxiliary
lemmas that are required for proving the main conjecture.

An alternative approach, proposed in the context of CHCs~\cite{De&18a}, 
consists in transforming a given set of clauses 
into a new set: (i) where all ADT terms are removed 
{(without introducing new function symbols)},
and (ii)~whose satisfiability implies the satisfiability 
of the original set of clauses.
This approach has the advantage of separating the concern of dealing
with ADTs (at transformation time) from the concern of dealing with
simpler, non-inductive constraint theories (at solving time), thus
avoiding the complex interaction between inductive reasoning and 
constraint solving.
It has been shown~\cite{De&18a} that the transformational approach 
compares well with induction-based tools
in the case where lemmas are not needed in the proofs.
However, in some satisfiability problems, if suitable lemmas are 
not provided, the transformation fails to remove the ADT terms.


The main contributions of this paper are as follows.

\noindent \hangindent=6mm 
(1) We extend the transformational approach by proposing a new rule, called 
{\it differential replacement}, based on the introduction of suitable 
{\it difference predicates},
which play a role similar to that of lemmas in inductive proofs.
We prove that the combined use of the fold/unfold transformation rules~\cite{EtG96}
and the differential replacement rule is {\it sound}, that is,
if the transformed set of clauses is satisfiable, then 
the original set of clauses is satisfiable.
	
\noindent \hangindent=6mm 
(2) We develop a transformation algorithm that removes ADTs from CHCs 
by applying the fold/unfold and the differential replacement
rules in a fully automated way. 
	
\noindent \hangindent=6mm 
(3)~{Due to the undecidability of the satisfiability problem for CHCs, 
our technique for ADT removal is incomplete. Thus, we evaluate  
its effectiveness from an experimental
point of view, and in particular we discuss the results obtained by the 
implementation of our technique in a tool, called {\sc AdtRem}}. 
We consider a set of CHC satisfiability problems on ADTs
taken from various benchmarks which are used for evaluating inductive theorem provers.
The experiments show that {\sc AdtRem} is competitive
with respect to Reynolds and Kuncak's tool that augments the CVC4 solver with  inductive reasoning~\cite{ReK15}.

\smallskip

\noindent
The paper is structured as follows. 
In Section~\ref{sec:IntroExample} we present an introductory, motivating example.
In Section~\ref{sec:CHCs} we recall some
basic notions about CHCs.
In Section~\ref{sec:TransfRules} we introduce the rules used
in our transformation technique and, in particular, the novel 
differential replacement rule, and we show their soundness.
In Section~\ref{sec:Strategy} we present a transformation algorithm
that uses the transformation rules for removing ADTs from sets of CHCs.
In Section~\ref{sec:Experiments} we illustrate the {\sc AdtRem} tool
and we present the experimental results we have obtained.
Finally, in Section~\ref{sec:RelConcl} we discuss the related work and
make some concluding remarks.

	\section{A Motivating Example}
	\label{sec:IntroExample}
	Let us consider the following functional program $\mathit{Reverse}$,
which we write using the OCaml syntax~\cite{Le&17}:


{\small
\begin{verbatim}
  type list = Nil | Cons of int * list;;
  let rec append l ys = match l with
    | Nil -> ys     | Cons(x,xs) -> Cons(x,(append xs ys));;
  let rec rev l = match l with
    | Nil -> Nil    | Cons(x,xs) -> append (rev xs) (Cons(x,Nil));;
  let rec len l = match l with
    | Nil -> 0      | Cons(x,xs) -> 1 + len xs;;
\end{verbatim}
}


\noindent
The  functions {\tt append}, {\tt rev}, and {\tt len}
compute list concatenation, list reversal, and list length, respectively.
Suppose we want to prove the following property:

{\small
$\mathtt{\forall}$ 
{\tt xs,ys.\ len\ (rev\ (append\ xs\ ys))\ =\ (len\ xs)\ +\ (len\ ys)} 
\hfill$(1)$\hspace*{10mm}
}

\vspace{.5mm}

\noindent
Inductive theorem provers construct a proof of
this property by induction on the structure of the list $\mathtt {l}$, by assuming 
the knowledge of the following lemma:

\vspace{.5mm}

{\small
$\mathtt{\forall}$ 
{\tt x,l. len (append\ l\ (Cons(x,Nil)))\ =\ (len\ l)\ +\ 1} 
\hfill $(2)$\hspace*{10mm}
}

\vspace{.5mm}

\noindent
The approach we follow in this paper avoids the explicit
use of induction principles and also the knowledge of {\it ad hoc} lemmas.
First, we consider the translation of Property~($1$) into a set of 
constrained Horn clauses~\cite{De&18a,Un&17}, as follows\footnote{In the examples,
we use Prolog syntax for writing clauses, instead of the more verbose
SMT-LIB syntax. The predicates {\tt \textbackslash=} (different from), {\tt =} (equal to), {\tt <} (less-than),
{\tt >=}~(greater-than-or-equal-to)
denote constraints between integers. The last argument of a Prolog 
predicate stores the value of the corresponding function.}:



{\small
\begin{verbatim}
1. false :- N2\=N0+N1, append(Xs,Ys,Zs), rev(Zs,Rs), 
            len(Xs,N0), len(Ys,N1), len(Rs,N2).
\end{verbatim}
\begin{verbatim}
2. append([],Ys,Ys).   3. append([X|Xs],Ys,[X|Zs]) :- append(Xs,Ys,Zs).
4. rev([],[]).         5. rev([X|Xs],Rs) :- rev(Xs,Ts), append(Ts,[X],Rs).
6. len([],N) :- N=0.   7. len([X|Xs],N1) :- N1=N0+1, len(Xs,N0).
\end{verbatim}
}


\noindent
%
%
%
%
The set of clauses {\tt 1}--{\tt 7} is satisfiable if and only if 
Property~$(1)$ holds. 
However, state-of-the-art CHC solvers, such as Z3 or Eldarica, 
fail to prove the satisfiability of clauses {\tt 1}--{\tt 7},
because those solvers do not incorporate any induction principle on lists.
Moreover, some tools that extend SMT solvers  
with induction~\cite{ReK15,Un&17} fail on this particular example
because they are not able to introduce Lemma~$(2)$.

To overcome this difficulty, we apply the transformational 
approach based on the fold/unfold rules~\cite{De&18a}, whose objective is to 
transform a given set of clauses into a new set without occurrences of list variables,
whose satisfiability can be checked by
using CHC solvers based on the theory of Linear Integer Arithmetic only.
The soundness of the transformation rules ensures that the satisfiability of the transformed clauses
implies the satisfiability of the original ones.
We apply the {\em Elimination Algorithm}~\cite{De&18a} as follows.
First, we introduce a new clause:


{\small
\begin{verbatim}
8. new1(N0,N1,N2) :- append(Xs,Ys,Zs), rev(Zs,Rs), 
                     len(Xs,N0), len(Ys,N1), len(Rs,N2).
\end{verbatim}
}


\noindent
whose body is made out of the atoms of clause~{\tt 1} which have at least
one list variable, and whose head arguments are 
the integer variables of the body.
By folding, \raisebox{-1.5mm}{\rule{0mm}{0mm}} 
from clause~1 we derive a new clause without occurrences 
of lists:\nopagebreak


{\small
\begin{verbatim}
9. false :- N2\=N0+N1, new1(N0,N1,N2).
\end{verbatim}
}


\noindent
We proceed by eliminating lists from clause~{\tt 8}. 
By unfolding clause~{\tt 8}, we replace some predicate calls by their 
definitions and we derive the two new clauses:


{\small
\begin{verbatim}
10. new1(N0,N1,N2) :- N0=0, rev(Zs,Rs), len(Zs,N1), len(Rs,N2).
11. new1(N01,N1,N21) :- N01=N0+1, append(Xs,Ys,Zs), rev(Zs,Rs), 
                len(Xs,N0), len(Ys,N1), append(Rs,[X],R1s), len(R1s,N21).
\end{verbatim}
}



\noindent
We would like to fold clause~{\tt 11} using clause~{\tt 8}, so as to derive a 
recursive definition of
{\tt new1} without lists.
Unfortunately, this folding step cannot be performed because
the body of clause~{\tt 11} does not contain a variant of
the body of clause~{\tt 8}, and hence
the Elimination Algorithm fails to eliminate lists in this example.

Thus, now we depart from the Elimination Algorithm and 
we continue our derivation by observing that
the body of clause~{\tt 11} contains the 
{\em subconjunction} `{\small{\tt append(Xs,Ys,Zs),} {\tt rev(Zs,Rs), len(Xs,N0), 
len(Ys,N1)}}' 
of the body
of clause~{\tt 8}.
Then, in order to find a variant of the whole body of clause~{\tt 8}, we may replace 
in clause~{\tt 11} the remaining subconjunction 
`{\small{\tt  append(Rs,[X],R1s), len(R1s,N21)}}'  by the new subconjunction
`{\small{\tt len(Rs,N2), diff(N2,X,N21)}}', where {\small{\tt diff}}
is a predicate, called {\em difference predicate}, 
defined as follows:


{\small
\begin{verbatim}
12. diff(N2,X,N21) :- append(Rs,[X],R1s), len(R1s,N21), len(Rs,N2).
\end{verbatim}
}


\noindent
From clause~{\tt 11}, by performing that replacement, 
we derive the following clause:\nopagebreak


{\small
\begin{verbatim}
13. new1(N01,N1,N21) :- N01=N0+1, append(Xs,Ys,Zs), rev(Zs,Rs), 
                len(Xs,N0), len(Ys,N1), len(Rs,N2), diff(N2,X,N21).
\end{verbatim}
}


\noindent
Now, we can fold clause~{\tt 13} using clause~{\tt 8}
and we derive a new clause without list arguments:


{\small
\begin{verbatim}
14. new1(N01,N1,N21) :- N01=N0+1, new1(N0,N1,N2), diff(N2,X,N21).  
\end{verbatim}
}


\noindent
At this point, we are left with the task of removing list arguments from 
clauses~{\tt 10} and~{\tt 12}.
As the reader may verify, this can be done by applying
the Elimination Algorithm without the need of introducing additional
 difference predicates.
By doing so, we get the following final set of clauses without list arguments:


{\small
\begin{verbatim}
false :- N2\=N0+N1, new1(N0,N1,N2).
new1(N0,N1,N2) :- N0=0, new2(N1,N2).
new1(N0,N1,N2) :- N0=N+1, new1(N,N1,M), diff(M,X,N2).
new2(M,N) :- M=0, N=0.
new2(M1,N1) :- M1=M+1, new2(M,N), diff(N,X,N1).
diff(N0,X,N1) :- N0=0, N1=1.
diff(N0,X,N1) :- N0=N+1, N1=M+1, diff(N,X,M).
\end{verbatim}
}


{The Eldarica CHC solver proves the satisfiability of this set of clauses by computing 
the following model 
(here we use a Prolog-like syntax):}


{\small
\begin{verbatim}
new1(N0,N1,N2) :- N2=N0+N1, N0>=0, N1>=0, N2>=0.
new2(M,N) :- M=N, M>=0, N>=0.
diff(N,X,M) :- M=N+1, N>=0.
\end{verbatim}
}

Finally, we note that if in clause~12  we substitute the atom
{\small{\tt diff(N2,X,N21)}} by its model computed by Eldarica, namely the constraint `{\small{\tt N21=N2+1, N2>=0}}', 
we get exactly the CHC translation of Lemma~$(2)$. 
Thus, in some cases, the introduction of the difference predicates can be
viewed as a way of automatically introducing the
lemmas needed when performing inductive proofs.



	\section{Constrained Horn Clauses}
	\label{sec:CHCs}
	Let \textit{LIA} be the theory of linear integer arithmetic and
\textit{Bool} be the theory of boolean values.
A {\em constraint} is a quantifier-free formula of $\textit{LIA}\cup\textit{Bool\/}$. 
Let~$\mathcal C$ denote the set of all constraints.
Let $\mathcal L$ be a typed first order language with equality~\cite{End72}
which includes the language of $\textit{LIA}\cup\textit{Bool\/}$.
Let $\textit{Pred}$ be a set of predicate
symbols in $\mathcal L$ not occurring in the language of $\textit{LIA}\cup\textit{Bool\/}$.


The integer and boolean types are said to be the {\it basic types}.
For reasons of simplicity we do not consider any other basic types,
such as real number, arrays, and bit-vectors, which are 
usually supported by SMT solvers~\cite{CVC4,DeB08,HoR18}.
We assume that all non-basic types 
are specified  by suitable data-type declarations (such as the 
{\small{\tt declare-datatypes}} declarations adopted by SMT solvers),
and are collectively called {\it algebraic data types} (ADTs). 

An {\it atom} is a formula of the form $p(t_{1},\ldots,t_{m})$,
where~$p$ is a typed predicate symbol in $\textit{Pred}$, and 
$t_{1},\ldots,t_{m}$ are typed terms constructed out of individual
variables, individual constants,  and function symbols.
A~{\it constrained Horn clause}  (or simply, a {\it clause}, or a CHC) is 
an implication of the form  
$A\leftarrow c, B$ (for clauses we use the logic programming notation, where
comma denotes conjunction). The conclusion (or {\it head\/}) $A$ 
is either an atom or \textit{false}, 
the premise (or {\it body\/}) is the conjunction of
a constraint  $c\in\mathcal C$, and a (possibly empty) conjunction~$B$ of atoms. 
If $A$ is an atom of the form $p(t_{1},\ldots,t_{n})$, the predicate $p$
is said to be a {\it head predicate}.
A clause whose head is an atom is called a {\it definite clause},
and a clause whose head is {\it false} is called a {\it  goal\/}.

We assume  that, for every atom $A$ which is the head of a clause, 
(i)~each term of basic type occurring in $A$ is a variable, and
(ii)~no variable of basic type occurs in $A$ more than once.
For instance, atom {\small{\tt p(X,[Y\,|\,T])}} may occur as a head, 
while by Condition~(i), the atoms {\small{\tt p(3,[Y\,|\,T])}}
and {\small{\tt p(X,[Y+Z\,|\,T])}} may not.
Conditions~(i) and (ii) on head atoms can always be enforced
at the expense of introducing new variables subject to constraints in the body of the clause.
These conditions ensure that, when applying the  unfolding rule (see Section~\ref{sec:TransfRules}),
the unification of terms of basic type
can be delegated to constraint solving.

We assume that all variables in a clause are universally quantified in front, and thus
we can freely rename those variables.
Clause $C$ is said to be a {\it variant}
of clause $D$ if $C$ can be obtained from $D$ by renaming variables 
and rearranging the order of the atoms in its body. 
Given a term~$t$, by ${\it vars}(t)$ we denote the set of all 
variables occurring in $t$.
Similarly {for the set of all free variables occurring in a formula}.
Given a formula $\varphi$ in ${\mathcal L}$, we denote by
 $\forall (\varphi)$ its {\it universal closure}. 



Let~$\mathbb D$ be the usual interpretation for the symbols in
$\textit{LIA}\cup\textit{Bool\/}$, and let a~\mbox{\it ${\mathbb D}$-interpretation} be an interpretation of~$\mathcal L$
that, for all symbols occurring in~$\textit{LIA}\cup\textit{Bool\/}$, agrees with
${\mathbb D}$.
%
%
%

A set $P$ of CHCs is {\it satisfiable} if
it has a ${\mathbb D}$-model and it is {\em unsatisfiable}, otherwise.
Given two {${\mathbb D}$-interpretations $\mathbb I$ and $\mathbb J,$
we say that $\mathbb I$ is {\em included} in $\mathbb J$
if for all ground atoms~$A$, $\mathbb I\models A$ implies $\mathbb J\models A$.}
Every set $P$ of definite  clauses is {satisfiable} and has a {\it least} 
(with respect {to inclusion}) ${\mathbb D}$-model,
denoted $M(P)$~\cite{JaM94}.
If $P$ is any set of constrained Horn clauses and
$Q$ is the set of the goals in~$P$, then we define
$\textit{Definite}(P) \!=\! P \setminus Q$.
We have that~$P$ is satisfiable
if and only if $M(\textit{Definite}(P))\models Q$.
%

We will often use tuples of variables 
as arguments of predicates and 
write $p(X,Y)$, instead of 
$p(X_{1},\ldots,X_{m}, Y_{1},\ldots,Y_{n})$, 
whenever the 
values of $m~(\geq 0)$ and $n~(\geq 0)$ are not relevant.
Whenever the order of the variables is not relevant, 
we will feel free to identify 
tuples of distinct variables 
with finite sets,
and we will extend to finite tuples the usual operations and relations 
 on finite sets.
Given two tuples $X$ and~$Y$ of distinct elements, 
(i)~their union $X\cup Y$ is obtained 
by concatenating them and removing all duplicated occurrences of elements, 
(ii)~their intersection $X\cap Y$ is obtained 
by removing from~$X$ the elements which do not occur in~$Y$,  
(iii)~their difference $X \!\setminus \!Y$ is obtained
by removing from $X$ the elements which occur in $Y$, and
(iv)~$X\!\subseteq\! Y$ holds if every element of $X$ occurs in $Y$.
For all $m\!\geq\!0$, equality of $m$-tuples is defined as follows:  
$(u_{1},\ldots,\!u_{m}) = (v_{1},\!\ldots,\!v_{m})$
iff $\bigwedge_{i=1}^{m}\!u_{i}\!=\!v_{i}$. \hspace{-1pt}The empty tuple \hspace{-1pt}$()$ is identified with the empty set $\emptyset$.


By $A(X,Y)$, where $X$ and $Y$ are disjoint tuples of {distinct} 
variables, we denote an atom $A$ such that $\mathit{vars}(A) = X\cup Y$.
Let $P$ be a set of definite clauses. We say that the atom $A(X,Y)$
is {\em functional from} $X$ {\em to} $Y$ {\em with respect to} $P$ if

(F1) $M(P) \models \forall X, Y, Z.\ A(X,Y) \wedge A(X,Z) ~\rightarrow~ Y\!=\!Z$

\noindent
The reference to the set $P$ of definite clauses is omitted, 
when understood from the context. Given a functional atom $A(X,Y)$, we say that $X$ and $Y$ are
its {\em input} and {\em output} (tuples of\/)
variables, respectively. 
The atom $A(X,Y)$ 
is said to be {\em total from} $X$ {\em to}~$Y$ {\em with respect to} $P$ if


(F2)  $M(P) \models \forall X \exists Y.\ A(X,Y)$

\noindent
If $A(X,Y)$ is a total, functional atom from $X$ to $Y$,
we may write $A(X;Y)$, instead of $A(X,Y)$.
For instance, {\small{\tt append(Xs,Ys,Zs)}} is a total, functional atom from 
{\small{\tt (Xs,Ys)}}
to {\small{\tt Zs}} with respect to the set of clauses 1--7 of 
Section~\ref{sec:IntroExample}.

Now we extend the above notions from atoms to conjunctions of atoms.
Let~$F$ be a conjunction 
$A_1(X_1;Y_1), \ldots,A_n(X_n;Y_n)$ such that: 
{(i)~$X\! = \!(\bigcup^n_{i=1} X_i)\!\setminus\!(\bigcup^n_{i=1} Y_i)$}, 
(ii)~$Y \! = \! (\bigcup^n_{i=1} Y_i)$, and
(iii)~for $i\!=\!1,\ldots,n,$ $Y_i$ is disjoint from $(\bigcup^i_{j=1} X_j) \cup (\bigcup^{i-1}_{j=1} Y_j)$. 
Then, the conjunction~$F$ is said to be a {\em \mbox{total}, functional conjunction from}~$X$ 
{\em to}~$Y$  and it is also written as~$F(X;Y)$. 
For $F(X;Y)$,
the above properties~(F1) and~(F2) hold
if we replace $A$ by~$F$.
For instance, {\small{\tt append(Xs,Ys,Zs), rev(Zs,Rs)}} is a total, functional conjunction
from {\small{\tt (Xs,Ys)}} to {\small{\tt (Zs,Rs)}} with respect 
to the set of clauses~1--7 of Section~\ref{sec:IntroExample}.

	
	\section{Transformation Rules for Constrained Horn Clauses}
	\label{sec:TransfRules}
	In this section we present the rules that we
propose for transforming CHCs, and in particular,
for introducing difference predicates, and
we prove their soundness.
We refer to Section~\ref{sec:IntroExample}
for examples of how the rules are applied.

First, we introduce the following notion of a {\it stratification} for a set of clauses. 
Let $\mathbb N$  denote the set of the natural numbers.
A \emph{level mapping} is a function 
$\ell\!:\mathit{Pred}\!\rightarrow\!\mathbb{N}$. For every predicate $p$,
the natural number $\ell(p)$  is said to be the {\it level\/} of~$p$.
Level mappings are extended to atoms by stating that 
the level $\ell(A)$ of an atom $A$ is the 
level of its predicate symbol.
A clause \( H\leftarrow c, A_{1}, \ldots, A_{n}
\) is {\it stratified with respect to}~$\ell$ if, for \( i\!=\!1,\ldots ,n \),
\(\ell (H)\geq \ell(A_i)\).
A set $P$ of CHCs is {\it stratified~with respect to $\ell$} if all clauses
in $P$ are stratified with respect to~$\ell$.
Clearly, for every set~$P$ of CHCs, there exists a level mapping $\ell$ such that
$P$ is stratified with respect to $\ell$~\cite{Llo87}. 

A {\it transformation sequence from} $P_{0}$ {\it to} $P_{n}$
is a sequence 
$P_0 \Rightarrow P_1 \Rightarrow \ldots \Rightarrow P_n$ of sets of CHCs 
such that, for $i\!=\!0,\ldots,n\!-\!1,$ $P_{i+1}$ is derived from $P_i$, denoted
$P_{i} \Rightarrow P_{i+1}$, by
applying one of the following \mbox{Rules~R1--R7}. We assume that $P_0$ is stratified with respect to~a given level mapping~$\ell$.

\smallskip
\noindent
(R1)~{\it Definition Rule.}  
Let $D$ be the clause $\textit{newp}(X_1,\ldots,X_k)\leftarrow c,A_1,\ldots,A_n$, 
where:
(i)~\textit{newp} is {a predicate symbol in $\textit{Pred\/}$ not occurring in 
	the sequence $P_0\Rightarrow P_1\Rightarrow\ldots\Rightarrow P_i$,}
\mbox{(ii)~$c$ is a constraint,} 
(iii)~the predicate symbols of $A_1,\ldots,A_n$
occur in $P_0$, and 
(iv)~$(X_1,\ldots,X_k)\subseteq \mathit{vars}(c,A_1,\ldots,A_n)$.
Then, $P_{i+1}= P_i\cup \{D\}$. We set $\ell(\textit{newp})=\textit{max}\,\{\ell(A_i) \mid i=1,\ldots,n\}$.

\smallskip

For $i\!=\!0,\ldots, n$, by $\textit{Defs}_i$ 
we denote the set of clauses, 
called {\it definitions}, 
introduced by Rule~R1 during the construction of $P_0\Rightarrow P_1\Rightarrow\ldots\Rightarrow P_i$. 
Thus, $\emptyset\!=\!\textit{Defs}_0\!\subseteq\!\textit{Defs}_1\!\subseteq\!\ldots.$
However, by using Rules R2--R7 we can replace a definition in $P_i$,
for $i\!>\!0$, and hence it may happen that $\textit{Defs}_{i+1}\!\not\subseteq\!P_{i+1}$.

\smallskip
\noindent
(R2)~{\it Unfolding Rule.} 
Let  $C$: $H\leftarrow c,G_L,A,G_R$ be a clause in $P_i$, where $A$ is an atom.
{Without loss of generality, we assume that
$\mathit{vars}(C)\cap\mathit{vars}(P_0)=\emptyset$.}
Let {\it Cls}: $\{K_{1}\leftarrow c_{1},
B_{1},~\ldots,~K_{m}\leftarrow c_{m}, B_{m}\}$, $m\geq 0$,
be the set of clauses in $P_0$,
such that: for $j=1,\ldots,m$,
(1)~there exists a most general unifier $\vartheta_j$ of $A$ and
$K_j$, and {(2)~the conjunction of constraints $(c, c_{j})\vartheta_j$ is satisfiable.}
Let $\mathit{Unf}(C,A,P_0) = \{(H\leftarrow  c, {c}_j,G_L, B_j, G_R) 
\vartheta_j \mid  j=1, \ldots, m\}$.
Then, by {\it unfolding~$C$ with respect to $A$}, we derive the set 
$\mathit{Unf}(C,A,P_0)$ of clauses and we get 
$P_{i+1}= (P_i\setminus\{C\}) \cup \mathit{Unf}(C,A,P_0)$.

\smallskip
When we apply Rule R2, we say that, for $j=1, \ldots, m,$ the atoms in the conjunction 
$B_j \vartheta_j$ are {\it derived} from $A$, and the atoms 
in the conjunction $(G_L, G_R) \vartheta_j$ are {\it inherited} from the corresponding atoms in the body of $C$.

\smallskip
\noindent
(R3)~{\it Folding Rule.} 
Let $C$: $H\leftarrow c, G_L,Q,G_R$ be a clause in $P_i$, 
and let
$D$: $K \leftarrow d, B$ be a clause in $\textit{Defs}_i$.
Suppose that:
(i)~either $H$ is $\mathit{false}$ or \mbox{$\ell(H) \geq \ell(K)$,} and
(ii)~there exists a substitution~$\vartheta$ such that~\mbox{$Q\!=\! B\vartheta $} and
$\mathbb D\models \forall(c \rightarrow d\vartheta)$.
By \textit{folding \( C\)
	using definition~\( D\)}, we derive clause 
\(E  \):~\( H\leftarrow c, G_L, K\vartheta, G_{R} \), and we get
\( P_{i+1}= (P_{i}\setminus\{C\})\cup \{E \} \).

\vspace{1.5mm}
\noindent
(R4)~{\it Clause Deletion Rule.} 
Let  $C$: $H\leftarrow c,G$ be a clause in $P_i$ such that the constraint~$c$ is unsatisfiable. Then, we get
$P_{i+1} = P_i \setminus\{C\}$.

\vspace{1.5mm}
\noindent
(R5)~{\it Functionality Rule.} 
Let $C$: $H\leftarrow c, G_L,F(X;Y),F(X;Z), G_R$ be a clause in~$P_i$,
where $F(X;Y)$ is a total, functional conjunction in 
$\mathit{Definite}(P_0)\cup \mathit{Defs}_i$.
By \textit{functionality}, from~$C$ we derive clause~$D$: $H\leftarrow c, Y\!=\!Z, G_L,F(X;Y),G_R$, and
we get \( P_{i+1}= (P_{i}\setminus\{C\})\cup \{D \} \).

\vspace{1.5mm}

\noindent
(R6)~{\it Totality Rule.} 
Let $C$: $H\leftarrow c,  G_L,F(X;Y),G_R$ be a clause in $P_i$
such that $Y \cap \mathit{vars}(H\leftarrow c,  G_L,G_R) = \emptyset$ and
$F(X;Y)$ is a total, functional conjunction in $\mathit{Definite}(P_0)\cup \mathit{Defs}_i$.
By \textit{totality}, from~$C$ we derive clause~$D$\,: $H\leftarrow c, G_L,G_R$
and we get \( P_{i+1}= (P_{i}\setminus\{C\})\cup \{D \} \).

\vspace{1.5mm}

{Since the initial set of clauses is obtained
by translating a terminating functional program,
the functionality and  totality properties hold by construction
and we do not need to prove them when we apply Rules~R5 and~R6.}


\vspace{1.5mm}
\noindent
(R7)~{\it Differential Replacement Rule.} 
Let $C$: $H\leftarrow c,  G_L,F(X;Y),G_R$ be a clause in $P_i$, and
let $D$: $\mathit{diff}(Z) \leftarrow d, F(X;Y), R(V;W)$
be a definition clause in $\mathit{Defs}_i$, where: 
(i)~$F(X;Y)$ and $R(V;W)$ are total, functional conjunctions 
with respect to~$\mathit{Definite}(P_0)\cup \mathit{Defs}_i$,
(ii)~$W \cap \vars(C)  = \emptyset$,  
(iii)~$\mathbb{D}\models\forall (c\rightarrow d)$, and
(iv)~$\ell(H)\!>\!\ell(\mathit{diff})$.
By {\it differential replacement},
we derive  
$E$: $H\!\leftarrow \!c, G_{\!L},R(V;\!W), \mathit{diff}(Z), G_{\!R}$
and we get $P_{i+1}= (P_{i}\setminus\{C\}) \cup \{E \}$. 

Note that no assumption is made on the set $Z$ of
 variables, apart from the one deriving from the fact that $D$ is a 
 definition, that is, 
 $Z\! \subseteq\! {\mathit{vars}}(d) \cup X\cup Y \cup V \cup W.$ 

{Rule R7 has a very general formulation that eases the proof of the Soundness Theorem (see Theorem~\ref{thm:unsat-preserv}),
which extends to Rules R1--R7 correctness results for transformations of 
 (constraint) logic programs~\cite{EtG96,Fi&04a,Sek09,TaS86}.
In the transformation algorithm of Section~\ref{sec:Strategy}, 
we will use a specific instance of Rule R7 which is sufficient for ADT removal (see, in particular, the Diff-Introduce step).}


\vspace{-1,5mm}
\begin{theorem}[Soundness]
\label{thm:unsat-preserv}
Let $P_0 \Rightarrow P_1 \Rightarrow \ldots \Rightarrow P_n$ be 
a transformation sequence
using Rules {\rm{R1--R7}}.
Suppose that the following condition holds\,$:$
	
\vspace{.5mm}
\noindent\hangindent=8mm
\makebox[8mm][l]{\rm \,(U)}for $i \!=\! 1,\ldots,n\!-\!1$, if $P_i \Rightarrow P_{i+1}$ by 
folding a clause in $P_i$ using a definition $D\!: H \leftarrow c,B$  
in $\mathit{Defs}_i$, 
then, for some $j\! \in\!\{1,\ldots,i\!-\!1,i\!+\!1,\ldots, n\!-\!1\}$, 
$ P_{j}\Rightarrow P_{j+1}$ by unfolding
$D$ with respect to an atom $A$ such that $\ell(H)=\ell(A)$. 
	
\vspace{.5mm}
\noindent
If	$P_n$ is satisfiable, then $P_0$ is satisfiable.
\end{theorem}\vspace*{-5mm}
\begin{proof}See Appendix. \end{proof}
\vspace{-2.5mm}

Thus, to prove the satisfiability of a 
set~$P_0$ of clauses,
it suffices: (i)~to construct a transformation sequence  
$P_0 \Rightarrow \ldots \Rightarrow P_n$, 
and then (ii)~to prove that $P_n$ is satisfiable.
Note, however, that if Rule~R7 is used,  it may happen that $P_0$ 
is satisfiable and $P_n$ is unsatisfiable, that is, some false counterexamples 
to satisfiability, so-called {\it false positives},
may be generated, as we now show.

\vspace{-1mm}
\begin{example}\label{ex:false-pos}
Let us consider the following set~$P_{1}$ of clauses derived by 
adding the definition clause {\small{\tt D}} to the initial 
set 
$P_{0}\!=\!\{${\small{\tt{C,1,2,3}}}\} of clauses:

\vspace*{-2mm}

{\small
\begin{verbatim}
C. false :- X=0, Y>0, a(X,Y).         
1. a(X,Y) :- X=<0, Y=0.       2. a(X,Y) :- X>0, Y=1.    3. r(X,W) :- W=1.
D. diff(Y,W) :- a(X,Y), r(X,W).
\end{verbatim}
}

\vspace{-2.5mm}

\noindent 
where: (i)~{\small{\tt{a(X,Y)}}} is a total, functional atom from {\small{\tt{X}}} to {\small{\tt{Y}}}, 
(ii)~{\small{\tt{r(X,W)}}}  is a total, functional atom from {\small{\tt{X}}} to {\small{\tt{W}}}, and 
(iii)~{\small{\tt{D}}} is a definition in~$\mathit{Defs}_{1}$. 
By applying Rule~R7, from $P_{1}$ we  derive the set $P_{2}\!=\!\{${\small{\tt{E,1,2,3,D}}}\}
of clauses where:

\vspace{-2mm}

{\small
\begin{verbatim}
E. false :- X=0, Y>0, r(X,W), diff(Y,W).
\end{verbatim}
}

\vspace{-2mm}

\noindent
Now we have that $P_{0}$ is satisfiable, while $P_{2}$ is unsatisfiable. 
\end{example}
\vspace*{-5mm}
	
	\section{An Algorithm for ADT Removal}
	\label{sec:Strategy}
	\vspace{-1mm}
Now we present  Algorithm~\Diff~for 
eliminating ADT terms from CHCs 
by using the transformation rules presented in Section~\ref{sec:TransfRules}
and automatically introducing suitable 
difference predicates.
If \Diff~terminates, it transforms a set {\it Cls} 
of clauses into a new set $\mathit{TransfCls}$
where the arguments of all predicates have basic type.
Theorem~\ref{thm:unsat-preserv} guarantees that if $\mathit{TransfCls}$ is satisfiable,
then also {\it Cls} is satisfiable.  

Algorithm~\Diff~(see Figure~\ref{fig:AlgoD}) 
removes ADT terms starting from the set~$\mathit{Gs}$ of 
goals in~$\mathit{Cls}$.
\Diff~collects these goals in~$\mathit{InCls}$ and stores in
$\mathit{Defs}$ the  definitions of new predicates
introduced by Rule R1.

\begin{figure}[!ht]
\vspace{-5mm}
\noindent \hrulefill\nopagebreak

\noindent {\bf Algorithm}~\Diff\\
{\em Input}: A set $\mathit{Cls}$ of clauses;
\\
{\em Output}: A set $\mathit{TransfCls}$ of clauses that have basic types.

\vspace*{-2mm}
\noindent \rule{2.0cm}{0.2mm}

\noindent 
Let $\mathit{Cls} = \mathit{Ds} \cup \mathit{Gs}$, where $\mathit{Ds}$ is a set of definite clauses and $\mathit{Gs}$ 
is a set of goals;

\noindent $\mathit{InCls}:=\mathit{Gs}$;
\noindent $\mathit{Defs}:=\emptyset$;
\noindent $\mathit{TransfCls}:=\emptyset;$

\noindent
{\bf while} $\mathit{InCls}\!\neq\!\emptyset$ {\bf do}

\hspace*{3mm}\vline\begin{minipage}{11cm} 
\makebox[4mm][l]{~$\scriptstyle\blacksquare$}$\mathit{Diff\mbox{-}Define\mbox{-}Fold}(\mathit{InCls},\mathit{Defs}, 
\mathit{NewDefs},\mathit{FldCls});$

\makebox[4mm][l]{~$\scriptstyle\blacksquare$}$\mathit{Unfold}(\mathit{NewDefs},\mathit{Ds},\mathit{UnfCls});$

\makebox[4mm][l]{~$\scriptstyle\blacksquare$}$\mathit{Replace}(\mathit{UnfCls}, \mathit{Ds}, \mathit{RCls});$

\makebox[5mm][l]{}$\mathit{InCls}:=\mathit{RCls};$~~
$\mathit{Defs}:=\mathit{Defs}\cup\mathit{NewDefs};$~~
$\mathit{TransfCls}:=\mathit{TransfCls}\cup\mathit{FldCls};$
\end{minipage} 

\vspace*{1mm} 
\noindent \hrulefill
\vspace*{-2mm} 
\caption{The ADT Removal Algorithm~\Diff \label{fig:AlgoD}.}
\vspace*{-6mm}
\end{figure}

Before describing the procedures used by Algorithm~\Diff,  
let us first introduce the following notions.

Given a conjunction $G$ of atoms, $\mathit{bvars}(G)$ 
(or $\mathit{adt\mbox{-}vars}(G)$)
denotes the set of variables in $G$ that have a basic type (or an ADT type, 
respectively).
We say that an atom (or clause) {\it has basic types} if {\em all\/} its 
arguments (or atoms, respectively)
have a basic type.
An atom (or clause) {\em has ADTs\/} if {\em at least one\/} of its arguments (or atoms, respectively) has an ADT type.

{Given a set (or a conjunction) $S$ of atoms, $\mathit{SharingBlocks}(S)$ 
denotes the partition of $S$ with respect to
the reflexive, transitive closure~$\Downarrow_S$ of the relation~$\downarrow_S$ defined as follows.
Given two atoms $A_1$ and $A_2$ in $S$, 
$A_1\! \downarrow_S\! A_2$ holds 
iff $\mathit{adt\mbox{-}vars}(A_1) \cap 
\mathit{adt\mbox{-}vars}(A_2)\!\neq\! \emptyset$.}
The elements of the partition are called the {\em sharing blocks} of~$S$.

A {\it generalization}  of a pair $(c_1,c_2)$ of constraints is a constraint $\alpha (c_1,c_2)$
such that $\mathbb D \models\forall (c_1 \rightarrow \alpha (c_1,c_2))$ and 
$\mathbb D \models\forall (c_2 \rightarrow \alpha (c_1,c_2))$~\cite{Fi&13a}.
In particular,
we consider the following generalization operator based on 
{\it widening}~\cite{CoH78}.
Suppose that $c_1$ is the conjunction $(a_1,\ldots,a_m)$ of atomic 
constraints, then $\alpha (c_1,c_2)$ is defined as the conjunction of 
all $a_i$'s in $(a_1,\ldots,a_m)$
such that \mbox{$\mathbb D\!\models\!\forall (c_2\! \rightarrow\! a_i)$}.
For any constraint $c$ and tuple $V$ of variables, the {\it projection} of
$c$ onto $V$ is a constraint $\pi(c,V)$ such that: 
(i)~$\mathit{vars}(\pi(c,V))\!\subseteq\! V$, and
(ii)~$\mathbb D \models \forall (c\!\rightarrow\!\pi(c,V))$. 
In our implementation,
$\pi(c,V)$ is computed from $\exists Y. c$, where $Y\!=\! \vars(c) \setminus V$,
by a quantifier elimination algorithm in the theory of booleans and {\it rational} 
numbers. This implementation is safe in our context, and avoids relying on
modular arithmetic, as is often done when eliminating quantifiers in LIA~\cite{Rab77}.

For two conjunctions $G_1$ and $G_2$ of atoms, $G_1\Embedded G_2$ holds
if $G_1\!=\!(A_1,\ldots,A_n)$ 
and there exists a subconjunction $(B_{1},\ldots, B_{n})$
of $G_2$ (modulo reordering) such that, for $i\!=\!1,\ldots,n,$ $B_i$ is an {instance} of $A_i$.
A conjunction $G$ of atoms is {\it connected} if it consists of a single sharing block.

\vspace{1mm}
\noindent
$\scriptstyle\blacksquare$ {\it Procedure $\mathit{Diff\mbox{-}Define\mbox{-}Fold}$} (see Figure~\ref{fig:Diff}).
At each iteration of the body of the {\bf for} loop, 
the $\mathit{Diff\mbox{-}Define\mbox{-}Fold}$ procedure removes the ADT terms occurring in
a sharing block $B$ of the body of a clause~$C\!:$ $H\!\leftarrow\! c,  B, G'$ of $\mathit{InCls}$. 
This is done by possibly introducing some new definitions (using Rule R1)
and applying the Folding Rule~R3. To allow folding, some applications
of the Differential Replacement Rule~R7  may be needed.
We have the following four cases.

\begin{figure}[!ht]
\noindent \hrulefill \nopagebreak

\noindent {\bf Procedure $\mathit{Diff\mbox{-}Define\mbox{-}Fold}(\mathit{InCls},\mathit{Defs}, 
	\mathit{NewDefs},\mathit{FldCls})$}
\\
{\em Input}\/: A set {\it InCls} of clauses and a set {\it Defs} of definitions;
\\
{\em Output}\/: A set {\it NewDefs} of definitions and a set $\mathit{FldCls}$ 
of clauses with basic types.

\vspace{-2mm}
\noindent \rule{2.0cm}{0.2mm}


\noindent $\mathit{NewDefs} := \emptyset; \ \mathit{FldCls}:= \emptyset$;

\noindent {\bf for} each clause $C$: $H\leftarrow c, G$ in $\mathit{InCls}$ {\bf do}

\noindent 
\hspace{3mm}{\bf if} $C$ has basic types {\bf then}
$\mathit{InCls}\! :=\! \mathit{InCls}\!\setminus\! \{C\}$;\ $\mathit{FldCls}:=\mathit{FldCls}\cup\{C\}$

\noindent 
\hspace*{3mm}{\bf else}
\vspace{1mm}

\hspace{6mm}
\vline\hspace{1mm}\begin{minipage}{11.4cm}
let $C$ be $H\leftarrow c, B, G'$ 
where $B$ is a sharing block in $G$ that contains at least one atom 
that has ADTs;
\\
$\bullet$ ({\bf Fold}) {\bf if} in  $\mathit{Defs}\cup \mathit{NewDefs}$ there is a (variant of) 
clause $D$: $\mathit{newp}(V) \If d, B$\\
\hspace*{3.7mm}such that $\mathbb D \models\forall (c \rightarrow d)$
{\bf then}
\underline{\Down fold} $C$ using $D$ and derive $E$: $H\!\leftarrow c, \mathit{newp}(V),\! G'$; 
\\
$\bullet$ ({\bf Generalize}) {\bf else if} in  $\mathit{Defs}\cup
\mathit{NewDefs}$ there is a (variant of a) clause \\
\hspace*{3.5mm}$\mathit{newp}(V) \If d, B$ and $\mathbb D \not\models\forall
 (c \rightarrow d)$
{\bf then}
\\ 
\hspace*{4.5mm}\vline\hspace{1.5mm}\begin{minipage}{10.8cm} 
\underline{\Down introduce definition} $\mathit{GenD}$: $\mathit{genp}(V) \If \alpha(d,c), B$; 
\\
\underline{\Down fold} $C$ using $\mathit{GenD}$ and derive $E$: $H\leftarrow c, \mathit{genp}(V), G'$; 
\\
$\mathit{NewDefs} := \mathit{NewDefs} \cup \{\mathit{GenD}\}$;
\end{minipage}\vspace{.5mm}
\\ 
$\bullet$ ({\bf Diff}-{\bf Introduce}) {\bf else if} in  $\mathit{Defs}\cup
\mathit{NewDefs}$ there is a (variant of a) clause \\
\hspace*{3.9mm} $D$:\,$\mathit{newp}(U) \If d, B'$ 
such that: (i) $\mathit{vars}(C)\! \cap\!\mathit{vars}(D)\!=\!\emptyset$, and 
(ii)~$B' \Embedded B$
{\bf then}
\\ 
\hspace*{4.5mm}\vline\hspace{1.5mm}\begin{minipage}{10.8cm} 

take a maximal subconjunction $M$ of $B$, if any, such that: 
\hangindent=3mm
\\
(i) $B\!=\!(M, F(X;Y))$, for some \Down connected conjunction $M$ and non-empty conjunction $F(X;Y)$,
(ii) $B'\vartheta=(M, R(V;W))$, for some substitution $\vartheta$ such that
$W \cap \vars(C) = \emptyset$, and 
(iii)~for every atom
$A$ in $\{F(X;Y),R(V;W)\}$, $\ell(H)>\ell(A)$; 

\underline{\Down introduce definition} \hangindent=15mm $\widehat{D}$: $\mathit{diff}(Z) \leftarrow \pi(c,X),F(X;Y),R(V;W)$
 
\hspace*{3mm}where $Z\!=\!\mathit{bvars}(F(X;Y),R(V;W))$;

$\mathit{NewDefs} := \mathit{NewDefs} \cup \{\widehat{D}\}$;

\underline{\Down replace} $F(X;Y)$ by $(R(V;W), \mathit{diff}(Z))$ in $C$, and derive  clause\\
\hspace*{3mm}\rule{0mm}{2.9mm}$C'$: $H\leftarrow c, M, R(V;W), \mathit{diff}(Z), G'$;

{\bf if} $\mathbb D \models\forall (c \rightarrow d\vartheta)$
               
\hspace*{3mm}{\bf then}~\underline{\Down fold} $C'$ using $D$ and derive $E$: $H\leftarrow c,\mathit{newp}(U\vartheta), \mathit{diff}(Z), G'$;

\hspace*{3mm}{\bf else} 
\hspace{.5mm}\vline\hspace{1.5mm}\begin{minipage}[t]{9.5cm}%
\underline{\Down introduce\,definition}\,$\mathit{GenD}$:\,$\mathit{genp}(U') \!\If\!\alpha(d\vartheta,c), B'\vartheta$ \\
\hspace*{3mm}where $U'\!=\!\mathit{bvars}(B'\vartheta)$;

\underline{\Down fold} $C'$ using $\mathit{GenD}$ and derive $E$: $H\leftarrow c, \mathit{genp}(U'), \mathit{diff}(Z), G'$; 

$\mathit{NewDefs} := \mathit{NewDefs} \cup \{\mathit{GenD}\}$;
\end{minipage} 
\end{minipage}

\noindent
$\bullet$ ({\bf Project}) {\bf else}

\hspace*{4.5mm}\vline\hspace{1.5mm}\begin{minipage}{10.8cm}
\underline{\Down introduce definition} $\mathit{ProjC}$: $\mathit{newp}(V) \If \pi(c,V), B$ 
\ where $V\!=\!\mathit{bvars}(B)$;

\underline{\Down fold} $C$ using $\mathit{ProjC}$ and derive  
clause $E$: $H\leftarrow c, \mathit{newp}(V), G'$; 

$\mathit{NewDefs} := \mathit{NewDefs} \cup \{\mathit{ProjC}\}$;
\end{minipage}
\end{minipage}	
\vspace{-1mm}

\noindent 
\hspace{3mm}$\mathit{InCls}\! :=\! (\mathit{InCls}\setminus \{C\}) \cup \{E\}$;
%
	
\nopagebreak \vspace*{-2mm} \noindent \hrulefill
\vspace*{-1mm}
\caption{The {\it Diff-Define-Fold} Procedure.
\label{fig:Diff}}
\vspace*{-4mm}
\end{figure}	

\vspace{1.5mm}
\noindent
$\bullet$ ({\bf Fold}). 
We remove the ADT arguments occurring in $B$ by folding $C$ 
using a definition $D$ introduced at a previous step.
Indeed, the head of each definition introduced by Algorithm~\Diff~is by construction a 
tuple of variables of basic type.

\vspace{1.5mm}
\noindent
$\bullet$ ({\bf Generalize}). We introduce a new definition
$\mathit{GenD}\!:$ $\mathit{genp}(V) \If \alpha(d,c), B$
whose constraint is obtained by generalizing $(d,c)$, where $d$ is the 
constraint occurring in an already available definition whose body is $B$.
Then, we remove the ADT arguments occurring in $B$ by folding $C$ using 
 $\mathit{GenD}$. 

\vspace{1.5mm}
\noindent
$\bullet$ ({\bf Diff}-{\bf Introduce}).
Suppose that $B$ {\em partially matches} the body of an available definition
$D$: $\mathit{newp}(U) \If d, B'$, that is, 
for some substitution $\vartheta$,
$B=(M, F(X;Y))$, and  $B'\vartheta=(M, R(V;W))$.
Then, we introduce a difference predicate through the new definition
$\widehat{D}$: $\mathit{diff}(Z) \leftarrow \pi(c,X),F(X;Y),R(V;W),$
where $Z\!=\!\mathit{bvars}(F(X;Y),R(V;W))$ and, 
by Rule~R7, we replace the conjunction
$F(X;Y)$ by $(R(V;W), \mathit{diff}(Z))$ in the body of $C$,
thereby deriving $C'$. 
Finally, we remove the ADT arguments in $B$ by folding $C'$ using either $D$ or 
a clause $\mathit{GenD}$ whose constraint is a generalization of the pair~$(d\vartheta,c)$
of constraints. 

\smallskip
The example of Section~\ref{sec:IntroExample} allows us to illustrate 
this (Diff-Introduce) case. With reference to that example,
clause~$C$: $H\!\leftarrow\! c, G$ that we want to fold is clause~{\tt 11},
whose body has the single sharing block 
$B$: `{\small{\tt append(Xs,Ys,Zs)\!,}} {\small{\tt rev(Zs,Rs)\!,}}\,{\small{\tt len(Xs,N0)\!,}}\,{\small{\tt len(Ys,N1)\!,}}\,{\small{\tt append(Rs,[X],R1s)\!,}}\,{\small{\tt len(R1s,N21)}}'. 
Block $B$ partially\,matches\,the\,body\,`{\small{\tt append(Xs,\!Ys,Zs)\!,\,rev(Zs,Rs)\!,\,len(Xs,N0)\!,}}\,{\small{\tt len(Ys,N1)\!,}} {\small{\tt len(Rs,N2)}}'\,of clause\,{\tt 8}\,of Section\,\ref{sec:IntroExample} which plays\,the role\,of
definition\,{$D$:\,$\mathit{newp}(U)\! \If\! d, B'$} in this example.
Indeed, we have that: 

\smallskip
\noindent
$M$= ({\small{\tt append(Xs,Ys,Zs)\!,} {\tt rev(Zs,Rs)\!, len(Xs,N0)\!, len(Ys,N1)}}),

\noindent
$F(X;Y)$\! =\! ({\small{\tt  append(Rs,[X],R1s)\!,\;len(R1s,N21)}}), where 
$X$={\small{\tt  (Rs,X)}}, $Y$={\small{\tt  (R1s,N21)}},

\noindent
$R(V;W)$\! =\! {\small{\tt len(Rs,N2)}}, where $V$=\,{\small{\tt (Rs)}},  $Y$=\,{\small{\tt (N2)}}.

\smallskip
\noindent
In this example, $\vartheta$ is the identity substitution. Morevover,
the condition on the level mapping~$\ell$ required in the 
{\it Diff-Define-Fold} Procedure of Figure~\ref{fig:Diff} 
can be fulfilled by stipulating that 
$\ell(${\small{\tt new1}}$)\!>\!\ell(${\small{\tt append}}$)$ and 
$\ell(${\small{\tt new1}}$)\!>\!\ell(${\small{\tt len}}$)$.
Thus, the definition $\widehat{D}$ 
to be introduced~is:

\vspace{-2mm}

{\small
	\begin{verbatim}
	12. diff(N2,X,N21) :- append(Rs,[X],R1s), len(R1s,N21), len(Rs,N2).
	\end{verbatim}
}
\vspace{-2mm}

\noindent
Indeed, we have that: (i)~the projection $\pi(c,X)$ is 
$\pi(${\small{\tt N01=N0+1}},\,{\small{\tt (Rs,X)}}$)$, that is, the empty conjunction,
(ii)~$F(X;Y),\ R(V;W)$ is the body of clause~{\tt 12},
and (iii)~the head variables {\small{\tt N2}}, {\small{\tt X}}, 
and {\small{\tt N21}} are the integer
variables in that body. 
Then, by applying Rule~R7, we replace in clause~{\tt 11} the 
conjunction `{\small{\tt  append(Rs,[X],R1s)\!, len(R1s,N21)}}'  by the new conjunction
`{\small{\tt len(Rs,N2)\!,}} {\small{\tt diff(N2,X,N21)}}', hence deriving clause $C'$, which is
clause~{\tt 13} of Section~\ref{sec:IntroExample}.
Finally, by folding clause~{\tt 13} using clause~{\tt 8}, we derive clause~{\tt 14} of Section~\ref{sec:IntroExample},
which has no list arguments.

\vspace{1.5mm}
\noindent
$\bullet$ ({\bf Project}). If none of the previous three cases apply, then
we introduce a new definition $\mathit{ProjC}$: $\mathit{newp}(V) \If \pi(c,V), B,$ 
\ where $V=\mathit{bvars}(B)$. Then,
we remove the ADT arguments occurring in $B$ by {folding} $C$ using $\mathit{ProjC}$.

\smallskip
The $\mathit{Diff\mbox{-}Define\mbox{-}Fold}$ procedure may
introduce new definitions with ADTs in their bodies, which are 
added to {\it NewDefs} 
and processed by the $\mathit{Unfold}$ procedure.
In order to present this procedure, we need the following notions.

The application of Rule R2 is controlled by marking some atoms in the body of a clause as {\it unfoldable}.
If we unfold with respect to~atom~$A$ clause $C$: \mbox{$H \!\leftarrow c, L, A, R$}
the marking of the clauses
in $\mathit{Unf(C,A,Ds)}$ is handled as follows: 
the atoms derived from $A$ are not marked 
as unfoldable and each atom~$A''$ inherited from an atom~$A'$ 
in the body of $C$ is marked as unfoldable
iff $A'$ is marked as unfoldable.  

An atom $A(X;Y)$ in a conjunction $F(V;Z)$ of atoms is said to be a {\it source atom} if
$X\!\subseteq\! V$. Thus, a source atom corresponds to an innermost 
function call in a 
given functional expression.
For instance, in clause~{\tt 1} of 
Section~\ref{sec:IntroExample}, the source atoms are
{\small{\tt append(Xs,Ys,Zs)}}, {\small{\tt len(Xs,N0)}}, and 
{\small{\tt len(Ys,N1)}}. Indeed, the 
body of clause~{\tt 1} corresponds to
{\small{\tt len(rev(append xs ys))}} {\tt =}\hspace*{-1.8mm}{\tt /} {\small{\tt (len xs)+(len ys)}}.

An atom $A(X;Y)$ in the body of clause $C$: \mbox{$H \leftarrow c, {L}, A(X;Y), {R}$}
is a {\it head-instance} with respect to~a set~{\it Ds} of clauses if, 
for every clause $K\leftarrow d, B$ in~{\it Ds} such that: 
(1)~there exists a most general unifier $\vartheta$ of $A(X;Y)$ and
$K$, and (2)~the constraint $(c, d)\vartheta$ is
satisfiable, we have that $X\vartheta\! =\! X$. Thus, the input variables of
$A(X;Y)$ are not instantiated by unification.
For instance, the atom {\small{\tt append([X|Xs],Ys,Zs)}} is a head-instance, while
{\small{\tt append(Xs,Ys,Zs)}} is not.

In a set {\it Cls}  of clauses, predicate $p$ {\it immediately depends on} predicate $q$,
if in\,{\it Cls} there is a clause of the form $p(\ldots) \leftarrow \ldots, q(\ldots), \ldots$
The {\it depends on} relation is the transitive closure of the {\it immediately depends on}
relation.
Let $\prec$ be a well-founded ordering on tuples of terms such that, for
all terms $t,u,$ if $t\!\prec\! u$, then, for all substitutions $\vartheta,$ 
$t\vartheta\!\prec\! u\vartheta$. A predicate $p$ is {\it descending} with respect to~$\prec$
if, for all clauses, $p(t;u) \leftarrow c,\, p_1(t_1;u_1),\ldots,p_n(t_n;u_n),$
for $i\!=\!1,\ldots,n,$ if $p_i$ depends on $p$ then $t_i\!\prec\! t$.
An atom is descending if its predicate is descending.
The well-founded ordering $\prec$ we use in our implementation is based on the {\it subterm}
relation and is defined as follows: $(x_1,\ldots,x_k)\!\prec\! (y_1,\ldots,y_m)$
if every $x_i$ is a subterm of some $y_j$ and there exists $x_i$ which is a strict subterm
of some $y_j$. For instance, the predicates {\small{\tt append}}, 
{\small{\tt rev}}, and {\small{\tt len}} in
the example of Section~\ref{sec:IntroExample} are all descending.

\vspace{1mm}
\noindent
$\scriptstyle\blacksquare$ {\it Procedure Unfold}  (see Figure~\ref{fig:unfoldProc})  repeatedly 
applies Rule~R2 in two phases.
In Phase~1, the procedure unfolds the clauses in $\mathit{NewDefs}$ with respect to~at least one source atom. 
Then, in Phase~2, clauses are unfolded with respect to~head-instance atoms. 
Unfolding is repeated only w.r.t~descending atoms.
The termination of the {\it Unfold} procedure is ensured by the fact that
the unfolding with respect to~a non-descending atom is done at most once in each phase. 	

\begin{figure}[!ht]
\noindent \hrulefill
	
	\noindent {\bf Procedure $\mathit{Unfold}(\mathit{NewDefs},\mathit{Ds},\mathit{UnfCls})$}
	\\
	{\em Input}\/: A set $\mathit{NewDefs}$ of definitions and a set $\mathit{Ds}$ of definite clauses;
	\\
	{\em Output}\/: A set $\mathit{UnfCls}$ of clauses.
	
	\vspace*{-2.5mm} \noindent \rule{2.0cm}{0.2mm}
	
	
	\noindent
	$\mathit{UnfCls} := \mathit{NewDefs}$; ~Mark as unfoldable a nonempty set of source atoms in the body of each clause of $\mathit{UnfCls}$;
	
\noindent \hangindent=2mm
	- {\bf while} there exists a clause $C$: $H \leftarrow c, {L}, A, {R}$\, in $\mathit{UnfCls}$, for some conjunctions $L$ and~$R$, such that $A$ is an
	unfoldable atom {\bf do}
	
\vspace*{0.5mm}
	
\noindent
\hspace{5mm}$\mathit{UnfCls}:=(\mathit{UnfCls}-\{C\})\cup \mathit{Unf(C,A,Ds)}$;
	
\vspace*{0.5mm}
	
\noindent {- Mark} as unfoldable all atoms in the body of each clause in $\mathit{UnfCls}$;
	
\noindent \hangindent=2mm
- {\bf while} there exists a clause $C$: 	\mbox{$H \leftarrow c, {L}, A, {R}$} in $\mathit{UnfCls}$,  for some conjunctions $L$ and~$R$,
	such that $A$ is a head-instance atom with respect to~{\it Ds} and $A$ is either unfoldable or descending {\bf do}
	
\vspace*{0.5mm}
	
	\noindent
	\hspace{5mm}$\mathit{UnfCls}:=(\mathit{UnfCls}-\{C\})\cup \mathit{Unf(C,A,Ds)}$;
	
\vspace{-2mm}
\noindent \hrulefill
\vspace{-2mm}
\caption{The {\it Unfold} Procedure.\label{fig:unfoldProc}}
\vspace*{-3mm}	
\end{figure}

\vspace*{1mm}
\noindent
$\scriptstyle\blacksquare$ {\it Procedure Replace} simplifies some clauses
by  applying Rules~R5 and~R6 as long as possible.
{\it Replace} terminates because each application of either rule decreases 
the number of atoms. 

Thus, each execution of the {\it Diff-Define-Fold}, 
{\it Unfold}, and {\it Replace} procedures terminates.
However, Algorithm~\Diff~might not terminate because 
new predicates may be introduced
by {\it Diff-Define-Fold} at each iteration of 
the {\bf while}-{\bf do} of~\Diff.
Soundness of~\Diff~follows from soundness of the transformation 
rules~{(see Appendix)}.

\vspace*{-1mm}

\begin{theorem}[Soundness of {Algorithm}~\Diff] 
\label{thm:soundness-AlgorithmR}
Suppose that {Algorithm}~\Diff~terminates 
for an input set $\mathit{Cls}$ of clauses, and let $\mathit{TransfCls}$
be the output set of clauses.
Then, every clause in $\mathit{TransfCls}$ has basic types, and
if $\mathit{TransfCls}$ is satisfiable, then $\mathit{Cls}$ is satisfiable.
\end{theorem}
\vspace*{-1mm}

{Algorithm}~\Diff~is not complete, in the sense that, even if $\mathit{Cls}$ is a 
satisfiable set of input clauses,
then \Diff~may not terminate or, due to the use of Rule R7, it may terminate with
an output set $\mathit{TransfCls}$ of unsatisfiable clauses, thereby generating a 
false positive (see Example~\ref{ex:false-pos} in
Section~\ref{sec:TransfRules}).
However, due to well-known undecidability results for the satisfiability problem 
of CHCs,
this limitation cannot be overcome, unless we restrict the class of clauses we 
consider.
The study of such restricted classes of clauses is beyond the scope of the present 
paper and,
instead, in the next section, we evaluate the effectiveness of {Algorithm}~\Diff~from
an experimental viewpoint. 
	
	\section{Experimental Evaluation}
	\label{sec:Experiments}
In this section we present the results of some experiments we have 
performed for assessing the effectiveness of our transformation-based 
CHC solving technique.
We compare our technique with the one proposed by Reynolds and Kuncak~\cite{ReK15},
which extends the  SMT solver CHC4 with inductive reasoning.

\smallskip
\noindent
{\it Implementation.} We have developed the 
{\sc AdtRem} tool  for ADT removal,
which is based on an implementation of Algorithm~\Diff~in
the VeriMAP system~\cite{De&14b}.


\smallskip
\noindent{\it Benchmark suite and experiments.}
Our benchmark suite consists of 169 verification problems
over inductively defined data structures, such as lists, queues, heaps, and trees,
which have been adapted from the benchmark suite considered by 
Reynolds and Kuncak~\cite{ReK15}. 
These problems come from benchmarks used by various theorem provers:
(i)~53~problems come from CLAM~\cite{IrB96},
(ii)~11~from HipSpec~\cite{Cl&13},
(iii)~63~from IsaPlanner~\cite{DiF03,Jo&10}, and
(iv)~42~from Leon~\cite{Su&11}. 
We have performed the following experiments,
whose results are summarized 
in Table~\ref{tab:evaluation}\hspace{.2mm}\footnote{The tool and the benchmarks 
are available at {\small{\url{https://fmlab.unich.it/adtrem/}}}.
}\!.

\noindent
(1) We have considered Reynolds and Kuncak's {\bf dtt} encoding of the verification problems,
where natural numbers are represented using the built-in SMT type {\it Int},
and we have discarded: (i)~problems that do not use ADTs, and 
(ii)~problems that cannot be directly represented in Horn clause format.
Since {\sc AdtRem} does not support higher order functions, 
nor user-provided lemmas,
in order to make a comparison between 
the two approaches on a level playing field,
we have replaced higher order functions by suitable first order instances and
we have removed all auxiliary lemmas from the input verification problems.
We have also replaced the
basic functions recursively defined over natural numbers,
such as the {\it plus} and {\it less-or-equal} functions, 
by LIA constraints.

\noindent
(2) Then, we have translated each verification problem into a set, call it $P$, of CHCs 
in the Prolog-like syntax supported by {\sc AdtRem}
by using a modified version of the SMT-LIB parser of the ProB system~\cite{LeB03}.
We have run Eldarica and Z3\,\footnote{More specifically, 
Eldarica v2.0.1 and Z3 v4.8.0\,with\,the\,Spacer\,engine~\cite{Ko&14}.}\hspace{-1.4mm}, %
which use no induction-based mechanism for handling ADTs,
to check the satisfiability of $P$.
Rows `$\mathrm{Eldarica}$' and `$\mathrm{Z3}$' show the number of solved problems,
that is, problems whose CHC encoding has been proved satisfiable.

\noindent
(3) We have run 
algorithm~\Diff~on $P$ to produce
a set~$T$ of CHCs without ADTs. 
Row `\Diff'\,reports the number of\,problems for which
Algorithm\,\Diff~terminates.

\noindent
(4) We have converted $T$ into the SMT-LIB format, and then
{we have} run Eldarica and Z3
for checking its satisfiability.
Rows `$\mathrm{Eldarica_{\,noADT}}$' and `$\mathrm{Z3_{\,noADT}}$'
report {outside round parentheses} the number of solved problems. 
There was only one false positive (that is, a satisfiable set 
{$P$ of clauses
transformed into an unsatisfiable set $T$}),
which we have {classified as} an {unsolved} problem.

\noindent
(5) In order to assess the improvements due 
{to the use of} the differential replacement rule
we have applied to $P$ a modified version, call it {$\mathcal{R}^{\circ}$}, of the ADT removal algorithm~\Diff~that
\textit{does not} introduce difference predicates,
that is, the  {\it Diff-Introduce} case of the {\it Diff-Define-Fold} Procedure of Figure~\ref{fig:Diff} 
is never executed.
The number of problems {for which $\mathcal{R}^{\circ}$} terminates 
and  the number of solved problems using Eldarica and Z3
are shown {within round parentheses}
in rows `\Diff', `$\mathrm{Eldarica_{\,noADT}}$', and `$\mathrm{Z3_{\/noADT}}$', respectively.

\noindent
(6) Finally, we have run  the {\bf cvc4+ig} 
configuration of the CVC4 solver 
extended with inductive reasoning~\cite{ReK15}
on the 169 problems in SMT-LIB format obtained at Step (1).
Row `CVC4+Ind' reports the number of solved problems.

\begin{table}[htbp]
\vspace{-1mm}
	\begin{center}
		\begin{tabular}{|@{\hspace{1mm}}l@{\hspace{2mm}}||r@{\hspace{2mm}}|r@{\hspace{2mm}}|r@{\hspace{2mm}}|r@{\hspace{2mm}}||r@{\hspace{2mm}}|}
			\hline
			~~& \multicolumn{1}{c|}{~~CLAM~~} & \multicolumn{1}{c|}{~~HipSpec~~} & \multicolumn{1}{c|}{~~IsaPlanner~~} & \multicolumn{1}{c||}{~~Leon~~} & \multicolumn{1}{c|}{~~Total~~} \\ \hline \hline
			{\it {number of problems}} & 53 & 11 & 63 & 42 & 169 \\ 			
			$\mathrm{Eldarica}$ & 0 & 2 & 4 & 9 & 15 \\  
			$\mathrm{Z3}$ & 6 & 0 & 2 & 10 & 18 \\  
			\hline
			&&&&& \\[-3.4mm]
			\Diff  & (18) 36 & (2) 4 & (56) 59 & (18) 30 & (94) 129 \\ 
			$\mathrm{Eldarica_{\,noADT}}$  & (18)  32 & (2)   4 & (56) 57 & (18)  29 & (94)  122 \\ 
			$\mathrm{Z3_{\,noADT}}$ & (18)  29 & (2) 3 & (55)  56 & (18)  26 & (93) 114 \\ 
			\hline
			&&&&& \\[-3.4mm]
			CVC4+Ind  & 17 & 5 & 37 & 15 & 74 \\   
			\hline			
		\end{tabular}
		
		\vspace{3mm}
		\caption{{\small {\it Experimental results.}
				For each problem we have set a timeout limit of 300 seconds.
				Experiments have been performed on an Intel Xeon CPU E5-2640 2.00GHz with 64GB RAM under CentOS.\Down
	}}\label{tab:evaluation}
\vspace{-2mm}		
	\end{center}
\end{table}
\vspace*{-1mm}



\noindent
{\it Evaluation of Results.} 
The results of our experiments
show that ADT removal considerably increases 
the effectiveness of CHC solvers
without inductive reasoning support.
For instance, Eldarica is able to solve {15} problems 
out of {169}, 
while it solves {122} problems after the removal of ADTs.
When using Z3, the improvement is slightly lower, but still very significant because the number of problems solved by Z3 rises from 18 to 114.
Note also that, when the ADT removal terminates (129 problems out of 169), 
the solvers are very effective
(about 95\% successful verifications for Eldarica 
and about 88\% for Z3).
The improvements specifically due to the use of the difference replacement rule
are demonstrated by the increase of the number of problems 
{for which} the ADT removal {algorithm} terminates 
(from 94 to 129), and 
of the number of problems solved  (from 94 to 122, for Eldarica, and
from 93 to 114, for Z3). 

{\sc AdtRem} compares  favorably to CVC4 extended with induction
(compare rows `$\mathrm{Eldarica_{\/noADT}}$' and  `$\mathrm{Z3_{\/noADT}}$' 
to row `CVC4+Ind').
Interestingly, 
the effectiveness of CVC4
may be increased if one extends the problem formalization with 
extra 
lemmas which may be used for proving the main conjecture.
Indeed, CVC4 solves 100 problems when auxiliary lemmas are added,
and 134 problems when, in addition, it runs on the 
{\bf dti} encoding, 
where  natural numbers are represented using both the  built-in type  {\it Int} 
and the ADT definition with the zero and successor constructors.
Our results show that in most cases {\sc AdtRem} needs neither 
those extra axioms nor
that sophisticated encoding.


Finally, in Table~\ref{tab:examples} we report some problems 
solved by {\sc AdtRem} with Eldarica
that are not solved by CVC4 with induction (using any encoding and auxiliary lemmas), 
or vice versa. For details, see 
{\small{\url{https://fmlab.unich.it/adtrem/}}}.


\begin{table}[htbp]
\begin{center}
\begin{tabular}{|@{\hspace{2mm}}l@{\hspace{2mm}}|@{\hspace{2mm}}l@{\hspace{2mm}}|}



%
\hline

{\it Problem} & {\it Property proved by {\sc AdtRem} and not by {\rm CVC4}~~~~~}  \\ \hline\hline
CLAM goal6\hspace*{1.7cm} &  $\forall x,y.\, \mathit{len (rev (append(x,y))) = len(x) + len(y)}$ \\ \hline
CLAM goal49 & $\forall x.\, \mathit{mem(x, sort (y)) \Rightarrow  mem(x, y)}$ \\ \hline
IsaPlanner goal52 &  $\forall n,l.\, \mathit{count(n,l) =  count(n, rev(l))}$  \\ \hline
IsaPlanner goal80 &   $\forall l.\, \mathit{sorted (sort(l))}$   \\ \hline
Leon heap-goal13 &   $\forall x,l.\, \mathit{len (qheapsorta (x,l)) = hsize(x) + len(l)}$\\ 

%
%
%
%
%
%
%
\hline 
\end{tabular}\vspace*{2mm}
\begin{tabular}{|@{\hspace{2mm}}l@{\hspace{2mm}}|@{\hspace{2mm}}l@{\hspace{2mm}}|}
 \hline
{\it Problem} & {\it Property proved by {\rm CVC4} and not by {\sc AdtRem}}  \\ \hline\hline
CLAM goal18 &  $\forall x,y.\, \mathit{rev(append(rev(x), y)) = append(rev(y),x)}$  \\ \hline
HipSpec rev-equiv-goal4 &  $\forall x,y.\, \mathit{qreva(qreva(x,y), nil) = qreva(y,x)}$     \\ \hline
HipSpec rev-equiv-goal6 &  $\forall x,y,z.\, \mathit{append(qreva(x,y), z) = qreva(x,append(y,z))}$    \\ \hline
\end{tabular}

\vspace{4mm}		
\caption{{\small A comparison between {\sc AdtRem} with Eldarica 
and CVC4 with induction. 
}}
\label{tab:examples}
\end{center}
\vspace{-6mm}
\end{table}
\vspace*{-2mm}

	\section{Related Work and Conclusions}
	\label{sec:RelConcl}
Inductive reasoning is supported, with different degrees of human
intervention, by many theorem provers, such as
ACL2~\cite{ACL2}, 
CLAM~\cite{IrB96},
Isabelle~\cite{IsaHOLBook02},
HipSpec~\cite{Cl&13},
Zeno~\cite{So&12}, and
PVS~\cite{Ow&92}.
The combination of inductive reasoning and SMT solving techniques
has been exploited  by many tools for program 
verification~\cite{Lei12,PhW16,ReK15,Su&11,Un&17,Ya&19}.

Leino~\cite{Lei12} integrates inductive reasoning into the Dafny program verifier
by implementing a simple strategy that rewrites user-defined 
properties that may benefit from induction
into proof obligation to be discharged by Z3.
The advantage of this technique is that it fully decouples inductive 
reasoning from SMT solving.
Hence, no extensions to the SMT solver 
are required.

In order to extend CVC4 with induction, Reynolds and Kuncak~\cite{ReK15}
also consider the rewriting of formulas that may take advantage from inductive 
reasoning, but this is done dynamically, during the proof search.
This approach allows CVC4 to perform the rewritings lazily,
whenever new formulas are generated during the proof search,
and to use the partially solved 
conjecture, to generate lemmas that may help in the proof of the initial conjecture.

The issue of generating suitable lemmas during inductive proofs has been also
addressed by Yang et al.~\cite{Ya&19} and implemented in~{\sc AdtInd}.
In order to conjecture new lemmas, their algorithm makes use of a 
syntax-guided synthesis strategy driven by a grammar, which is 
dynamically generated from user-provided templates and the function 
and predicate symbols encountered during the proof search.
The derived lemma~conjectures are then checked by the  SMT solver Z3.

In order to take full advantage of the efficiency of SMT solvers
in checking satisfiability of quantifier-free formulas over LIA,
ADTs, and finite sets, the Leon verification system~\cite{Su&11} implements an
SMT-based solving algorithm to check the satisfiability of formulas 
involving recursively defined first-order functions.
The algorithm interleaves the unrolling of recursive functions
and the SMT solving of the formulas generated by the unrolling.
Leon can be used to prove properties of Scala programs with ADTs and 
integrates with the Scala compiler and the SMT solver Z3.
A refined version of that algorithm, restricted to
{\it catamorphisms}, 
has been implemented into a solver-agnostic tool, called RADA~\cite{PhW16}.

In the context of CHCs, Unno et al.~\cite{Un&17} have
proposed a proof system that combines inductive theorem proving with
SMT solving. This approach  
uses \mbox{Z3-PDR~\cite{HoB12}}
 to discharge proof obligations generated by the proof system,
and has been applied 
to prove relational properties 
of OCaml  programs.

The distinctive feature of the 
technique presented in this paper is that it does not make
use of any explicit inductive reasoning, but it follows a transformational 
approach.
First, the problem of verifying the validity of a universally
quantified formula over ADTs is reduced to the problem of checking the 
satisfiability of a set of CHCs. 
Then, this set of CHCs is transformed with the aim of
deriving a set of CHCs  {over basic types (i.e., integers) 
only}, whose satisfiability
implies the satisfiability of the original set.
In this way, the reasoning on ADTs is separated from
the reasoning on satisfiability, which can be performed
by specialized engines for CHCs on basic types
(e.g. Eldarica~\cite{HoR18} and Z3-Spacer~\cite{Ko&13}).
Some of the ideas presented here have been explored 
in~\cite{De&19b,De&19c},  
but there
neither formal results nor an automated strategy were presented.

A key success factor of our technique is the introduction of
difference predicates, which can be viewed
as the transformational counterpart of lemma generation.
Indeed, as shown in Section~\ref{sec:Experiments}, the use of difference predicates
greatly increases the power of CHC solving with respect to previous techniques based on 
the transformational approach, which do not use difference predicates~\cite{De&18a}.

As future work, we plan to apply our transformation-based verification technique
to more complex program properties, such as relational properties~\cite{De&16c}. 


%
%
		
	\bibliographystyle{abbrv}

	\newpage
	\section{Appendix}
    \newcommand{\Iff}{\leftrightarrow}
\newcommand{\ONLYIF}{\Rightarrow}
\newcommand{\IFF}{\Leftrightarrow}

In this appendix we show the proofs of the results presented in Sections~\ref{sec:TransfRules}
and~\ref{sec:Strategy}. 
First, we recall some definitions and facts from the literature~\cite{PeP08}.
{The least \mbox{$\mathbb D$-model} of a set~$P$ of clauses is the set, denoted 
$M(P)$, of all ground atoms which are true in $P$~\cite{JaM94}.}

A {\em reverse-implication-based transformation sequence}
is a sequence $P_0 \Rightarrow P_1 \Rightarrow \ldots \Rightarrow P_n$ of 
sets of clauses where, for $i\!=\!0,\ldots,n\!-\!1,$ $ P_{i+1}$ is derived 
from~$P_i$ by
applying one of the following rules (see Section~\ref{sec:TransfRules}): 
Definition (Rule~R1), 
Unfolding (Rule~R2), 
Folding (Rule~R3), and the following rule, called
{\it Body Weakening} (Rule~W).

\medskip

\noindent 
\textrm{(Rule~W)} \textit{Body Weakening.}~~Let  $C$:
${H}\leftarrow c, c_{1}, {G}_{L},
{G}_{1}, {G}_{R}$ be a clause in ${P}_{i}$, and suppose that the
following condition holds for some constraint~$c_{2}$ and conjunction~$G_{2}$ 
of atoms:

\smallskip

$M(\textit{Definite}(P_0)\cup \mathit{Defs}_i) \models \forall\,( c_{1}\! \wedge \! {G}_{1}
\rightarrow \exists Y.\, c_{2}\! \wedge \! {G}_{2})$ 

\smallskip
\noindent
where $Y= \mathit{vars}(\{c_{2},{G}_{2}\}) \setminus \mathit{vars}(\{{H}, c,
{G}_{L},{G}_{R}\})$. 
Suppose also that $\ell(H)>\ell(\mathit{A})$, for every atom $A$ occurring 
in $G_2$ and not in $G_1$.
By \emph{body weakening}, from
clause~$C$ we derive  clause~$D$:
${H}\leftarrow c, c_{2}, {G}_{L},
{G}_{2}, {G}_{R}$, and we get 
$P_{i+1} = ({P}_{i}\setminus\{C\})\cup \{D\}$.

\begin{theorem}
\label{thm:cons} If $P_0 \repl \ldots \repl P_n$
is a reverse-implication-based transformation sequence for which
Condition {\rm (U)} of Theorem~$\ref{thm:unsat-preserv}$ holds.
For \mbox{$i=0,\ldots,n$,} let $\mathcal{D}_i=\textit{Definite}(P_i)$.
Then $M(\mathcal{D}_0\cup \mathit{Defs}_n)\subseteq M(\mathcal{D}_n)$.
\end{theorem}

The proof of this theorem~\cite{PeP08} is based on the fact
that, for
$i=0,\ldots,n\!-\!1$, $M(\mathcal{D}_0\cup \mathit{Defs}_i) \models \mathcal{D}_i 
\leftarrow \mathcal{D}_{i+1}$,
hence the term {\em reverse-implication-based} transformation sequence. 
Note that, in particular, if we apply the body weakening rule to a clause $C$ of a 
given set $P_{i}$ of clauses whereby
we replace a conjunction $G_{1}$ in the body of $C$ by a 
new a conjunction $G_{2}$ such that $M(\mathcal{D}_0\cup \mathit{Defs}_i) 
\models \forall\,(G_{1} \rightarrow G_{2})$, then we derive a new 
clause $D$ 
such that
$M(\mathcal{D}_0\cup \mathit{Defs}_i) \models C \leftarrow D$.


\begin{theorem}
\label{thm:unsat} Suppose that $P_0 \repl \ldots \repl P_n$
is a reverse-implication-based transformation sequence and 
Condition~{\rm (U)} of
Theorem~$\ref{thm:unsat-preserv}$ holds.
If $P_n$ is satisfiable, then $P_0$ is satisfiable.
\end{theorem}

\begin{proof}
First, we observe that $P_0$ is satisfiable iff 
$P_0 \cup \textit{Defs}_n$ is satisfiable. Indeed, we have that: 
(i)~if~$\mathcal{M}$ is a $\mathbb D$-model of $P_0$, then the $\mathbb D$-interpretation
$\mathcal{M} \cup \{\textit{newp}(a_1,\ldots,a_k) \mid \textit{newp}$ 
is a head predicate
in $\textit{Defs}_n \mbox{ and } a_1,\ldots, a_k$ are ground terms$\}$
is a $\mathbb D$-model of $P_0 \cup \textit{Defs}_n$,
and (ii)~if $\mathcal{M}$ is a $\mathbb D$-model of $P_0 \cup \textit{Defs}_n$, then
all clauses of $P_0$ are true in $\mathcal{M}$, and hence
$\mathcal{M}$  is a $\mathbb D$-model of $P_0$.

Then, let us consider a new sequence  $P'_0\Rightarrow \ldots \Rightarrow P'_n$
obtained from the transformation sequence $P_0\Rightarrow \ldots \Rightarrow P_n$
by replacing each occurrence of \textit{false} in the head of
a clause by a fresh, new predicate symbol, say $f$.
$P'_0,\ldots,P'_n$ are sets of definite clauses, and thus, for $i=0,\ldots,n,$ $\textit{Definite}(P'_i)= P'_i$.
The sequence \mbox{$P'_0\Rightarrow \ldots \Rightarrow P'_n$} satisfies the 
hypotheses of Theorem~\ref{thm:cons}, and hence 
$M(P'_0)\cup \mathit{Defs}_n\!\subseteq M(P'_n)$.
We have that: 

\noindent
\makebox[12mm][l]{}$P_n$ is satisfiable

\noindent
\makebox[12mm][l]{implies}{$P'_n\cup \{\neg f\}$ is satisfiable}

\noindent
\makebox[12mm][l]{implies}{$f\not\in M(P'_n)$}

\noindent
\makebox[12mm][l]{implies,}\,by Theorem~\ref{thm:cons}, 
$f\not\in M(P'_0)\cup \textit{Defs}_n$

\noindent
\makebox[12mm][l]{implies}{$P'_0 \cup \textit{Defs}_n\cup \{\neg f\}$ is satisfiable}

\noindent
\makebox[12mm][l]{implies}{$P_0 \cup \textit{Defs}_n$ is satisfiable}

\noindent
\makebox[12mm][l]{implies}{$P_0$ is satisfiable.} \hfill$\Box$
\end{proof}

Now, (i)~in order to prove Theorem~\ref{thm:unsat-preserv} of 
Section~\ref{sec:TransfRules}, which states the soundness of Rules R1--R7, 
we show that Rules R4--R7 are all instances
of Rule~W.

\medskip

An application of Rule~R4~(Clause Deletion) whereby 
clause~$C$: $H\leftarrow c,G$ is deleted
whenever the constraint~$c$ is unsatisfiable, is equivalent to
the replacement of the body of clause $C$ by {\it false}. 
Indeed, if $c$ is unsatisfiable, we have that:~\nopagebreak

\smallskip
$M(\textit{Definite}(P_0)\cup \mathit{Defs}_i) \models \forall\,(c \wedge G \rightarrow {\it false})$

\smallskip
\noindent
and the applicability condition of Rule~W
trivially holds. Also the condition:

\smallskip
$\ell(H)>\ell(\mathit{A})$, for every atom $A$ in ${\it false}$ 

\smallskip
\noindent
trivially holds, because there are no atoms in ${\it false}$.
Thus, the replacement of the body of clause $H\leftarrow c,G$ by {\it false}
can be performed by applying Rule~W.

\medskip

Let us now consider Rule~R5 (Functionality). 
Recall that by $F(X;Y)$ we denote a conjunction of atoms 
that defines a total, functional relation from~$X$ to~$Y$.
When Rule~R5 is applied whereby a conjunction $F(X;Y),F(X;Z)$
is replaced by $Y\!=\!Z, F(X;Y)$, it is the case that 

\smallskip
$M(\mathit{Definite}(P_0)\cup \mathit{Defs}_i) \models \forall(F(X;Y) \wedge F(X;Z) \rightarrow Y\!=\!Z)$

\smallskip
\noindent
holds. When this replacement is performed, also the applicability 
condition of Rule~W on the levels of atoms trivially holds, and thus
Rule~R5 is an instance of Rule~W. 

\medskip

An application of Rule~R6 (Totality) replaces a conjunction $F(X;Y)$
by {\it true} (that is, the empty conjunction), which is implied by any formula. Hence,  Rule~R6 is an instance of Rule~W.

\medskip

Rule~R7 (Differential Replacement), is an instance of Rule~W as a 
consequence of the following lemma.

\begin{lemma}\label{lemma:R7unsat} Let us consider a transformation sequence 
$P_{0}\Rightarrow \ldots \Rightarrow P_{i}$ and a 
clause $C${\rm :}~$H\leftarrow c,  G_L,F(X;Y),G_R$ in $P_{i}$. 
Let us assume that by applying Rule~{\rm R7} on clause~$C$ using 
the definition clause

\smallskip
$D{\rm :}~\mathit{diff}(Z) \leftarrow d, F(X;Y), R(V;W),$

\smallskip
\noindent
where\,{\rm :} $(i)$~$W \cap \vars(C) = \emptyset,$  and
$(ii)$~$\mathbb{D}\models\forall (c\rightarrow d),$ 
we derive clause

\smallskip
$E${\rm :} $H\leftarrow c,  G_L,R(V;W), \mathit{diff}(Z), G_R$

\smallskip
\noindent
and we get the new set 
$P_{i+1} = (P_{i}\setminus\{C\}) \cup \{E \}$ of clauses. Then,

\smallskip
$M(\textit{Definite}(P_0) \cup \mathit{Defs}_i)\models \forall( c \wedge F(X;Y) \rightarrow \exists W.\,(R(V;W) \wedge \mathit{diff}(Z)))$.
\end{lemma}

\begin{proof} Let $\mathcal M$ denote $M(\textit{Definite}(P_0) 
\cup \mathit{Defs}_i)$.
Since $R(V;W)$ is a total, functional conjunction from $V$ to $W$
with respect to~$\textit{Definite}(P_0) \cup \mathit{Defs}_i$,
we have:

\smallskip

$\mathcal M\models \forall\,( c \wedge F(X;Y) \rightarrow \exists W.\, R(V;W))$ \hfill (1)~~~~

\smallskip
\noindent
Since, by Condition~$(i)$, none of the variables in $W$ occurs in~$C$, from definition~$D$, we get:

\smallskip

$\mathcal M \models \forall\,(d \wedge F(X;Y) \wedge R(V;W) \rightarrow \mathit{diff}(Z))$\hfill (2)~~~~

   \smallskip

\noindent
From~(1) and~(2), by Condition~$(ii)$, we get the thesis.\hfill $\Box$
\end{proof}

\medskip

From Lemma~\ref{lemma:R7unsat}
it follows that Rule~R7, which replaces in the body of 
clause~$C{\rm :}~H\leftarrow c,  G_L,F(X;Y),G_R$
the conjunction~$F(X;Y)$ by the new conjunction~$R(V;W), \mathit{diff}(Z)$, 
whenever Condition $\ell(H)\!>\!\ell(\mathit{diff})$ holds, is an instance of Rule~W.

Having proved Lemma~\ref{lemma:R7unsat}), we have that Rules R4--R7 
are all instances of Rule~W. Thus, we get the following lemma.

\begin{lemma}
\label{lemma:R1-R7-rev-impl} 
Suppose that $P_0 \repl \ldots \repl P_n$ 
is a transformation sequence using Rules \mbox{{\rm{R1}}--{\rm{R7}}}, and suppose also that
Condition {\rm (U)} of Theorem~$\ref{thm:unsat-preserv}$ holds.
Then $P_0 \repl \ldots \repl P_n$
is a reverse-implication-based transformation sequence.
\end{lemma}

Now Theorem~\ref{thm:unsat-preserv} of Section~\ref{sec:TransfRules}
follows from Lemma~\ref{lemma:R1-R7-rev-impl}
and Theorem~\ref{thm:unsat}. \hfill~$\Box$

\medskip

Finally, we prove the soundness of Algorithm~\Diff, that is, Theorem~\ref{thm:soundness-AlgorithmR} of 
Section~\ref{sec:Strategy}.

Each procedure used in {Algorithm}~\Diff~consists of a sequence of 
applications of Rules R1--R7.
Thus, the thesis follows from Theorem~\ref{thm:unsat-preserv}, if Condition~(U)
of the hypothesis of that theorem holds. 
This condition is satisfied if for each application of the {\it Unfold} procedure,
the nonempty set of source atoms that are marked as unfoldable in the body of a clause,
includes at least one atom with the same level as the head of the clause.
This property can always be enforced by dynamically constructing the level
mapping~$\ell$ during the execution of Algorithm~\Diff.
\hfill~$\Box$

\end{document}